\documentclass[letterpaper]{article} 
\usepackage{aaai24}  
\usepackage{times}  
\usepackage{helvet}  
\usepackage{courier}  
\usepackage[hyphens]{url}  
\usepackage{graphicx} 
\usepackage{natbib}  
\usepackage{caption} 
\usepackage{algorithm}
\usepackage{algorithmic}

\usepackage{newfloat}
\usepackage{listings}
\DeclareCaptionStyle{ruled}{labelfont=normalfont,labelsep=colon,strut=off} 
\floatstyle{ruled}
\newfloat{listing}{tb}{lst}{}
\floatname{listing}{Listing}

\usepackage{multirow}

\usepackage{amssymb}
\usepackage{amsmath}
\usepackage{amsthm}

\newtheorem{lemma}{Lemma}
\newtheorem{theorem}{Theorem}

\newtheorem{definition}{Definition}
\newtheorem{example}{Example}

\title{Approximate Integer Solution Counts over Linear Arithmetic Constraints}
\author{
    Cunjing Ge\textsuperscript{\rm 1, \rm 2}
}
\affiliations{
    \textsuperscript{\rm 1}National Key Laboratory for Novel Software Technology, Nanjing University, China \\
    \textsuperscript{\rm 2}School of Artificial Intelligence, Nanjing University, China \\
    gecunjing@nju.edu.cn
}

\usepackage{bibentry}

\begin{document}

\maketitle

\begin{abstract}
Counting integer solutions of linear constraints has found interesting applications in various fields.
It is equivalent to the problem of counting lattice points inside a polytope.
However, state-of-the-art algorithms for this problem become too slow for even a modest number of variables.
In this paper, we propose a new framework to approximate the lattice counts inside a polytope
with a new random-walk sampling method.
The counts computed by our approach has been proved approximately bounded by a $(\epsilon, \delta)$-bound.
Experiments on extensive benchmarks show that our algorithm could solve polytopes with dozens of dimensions,
which significantly outperforms state-of-the-art counters.
\end{abstract}

\section{Introduction}

As one of the most fundamental type of constraints,
linear constraints (LCs) have been studied thoroughly in many areas.
In this paper, we consider the problem of counting approximately the number of integer solutions of a set of LCs.
This problem has many applications,
such as counting-based search\,\cite{ZanariniP07,Pesant16}, simple temporal planning~\cite{HuangLOB18},
probabilistic program analysis\,\cite{GeldenhuysDV12,LuckowPDFV14}, etc..
It also includes as a special case several combinatorial counting problems that have been studied, like that
of estimating the permanent of a matrix\,\cite{JerrumS89,GamarnikK10,HarviainenRK21},
the number of contingency tables\,\cite{CryanDGJM02,DesalvoZ20},
solutions to knapsack problems\,\cite{DyerFKKPV93}, etc..
Moreover, it can be incorporated as a subroutine for \#SMT\,(LA)\,\cite{GeMZZ18}.
Since a set of LCs represents a convex polytope,
its integer solutions correspond to lattice points inside the polytope.
Accordingly, we do not distinguish the concepts of polytopes and sets of LCs in this paper.

It is well-known that counting lattice points in a polytope is \#P-hard\,\cite{Valiant79}.
On the implementation front, the first practical tool for lattice counting is \textsc{LattE}~\cite{LoeraHTY04},
which is an implementation of Barvinok's algorithm\,\cite{Barvinok93a,Barvinok94}.
The tool \textsc{barvinok}\,\cite{VerdoolaegeSBLB07} is the successor of \textsc{LattE}
with an in general better performance.
In practice, it often still has difficulties when the number of variables is greater than $10$ (preventing many applications).
The relation between the number of lattice points inside a polytope and the volume of a polytope has been studied for approximate integer counting\,\cite{GeMMZHZ19}.
However, it is inevitable that the approximation bounds may far off from exact counts.
A more recent work\,\cite{GeB21} introduced factorization preprocessing techniques to reduce polytopes dimensionality,
which are orthogonal to lattice counting,
they also require polytopes in specific forms.

An algorithm for sampling lattice points in a polytope was introduced in\,\cite{KannanV97},
which can be used to approximate the integer solution count,
though we are not aware of any implementation.
Since then, there have been a lot of works about sampling real points, such as Hit-and-run method\,\cite{Lovasz99,LovaszV06b}, and approximating polytopes' volume\,\cite{LovaszD12,CousinsV15,CousinsV18}.
As a result, the state-of-the-art volume approximation algorithms could solve general polytopes around $100$ dimensions.
Naturally, we wonder if they could be extended to integer cases.

The primary contribution of this paper is a novel approximate lattice counting algorithm,
in detail, it includes new methods with theoretical results as follows.
\begin{itemize}
\item A lattice sampling method is introduced,
    which is a combination of Hit-and-run random walk and rejection sampling.
    We proved that it generates samples in distribution limited by Hit-and-run method, which is nearly uniform.
\item A dynamic stopping criterion is proposed, which could be calculated by variance of approximations while running.
    We proved that errors of outputs approximately lie in $[1-\epsilon, 1+\epsilon]$ with probability at least $1 - \delta$,
    given $\epsilon, \delta$.
\end{itemize}
We evaluated our algorithm on an extensive set of random and application benchmarks.
We not only compared our tool with integer counters, but also with \#SAT counters by
translating LCs into propositional logic formulas.
Experimental results show that our approach scales to polytopes up to $80$ dimensions, which significantly outperforms the state-of-the-art counters.
We also observe that counts computed by our algorithm are bounded well by theoretical guarantees.


\section{Background}\label{sect:bg}
In this section, we first present definitions of notations,
and then briefly describe the sampling and volume approximation algorithms which inspired us.

\subsection{Notations and Preliminaries}\label{sect:prelim}

\begin{definition}
A linear constraint is an inequality of the form $a_1x_1 + \dots + a_nx_n\ \mathsf{op}\ b$,
where $x_i$ are numeric variables, $a_i$ are constant coefficients, and $\mathsf{op} \in \{ <, \le, >, \ge, = \}$.
\end{definition}

Without loss of generality, a set of linear constraints can be written in the form of: $\{A\vec{x} \le \vec{b}\}$,
where $A$ is a $m \times n$ coefficient matrix and $\vec{b}$ is a $1 \times n$ constant vector.
In the view of geometry, a linear constraint is a halfspace,
and a set of linear constraints is an $n$-dimensional polytope.

\begin{definition}
An $n$-dimensional polytope is in the form of
$$P = \{\vec{x}\in \mathbb{R}^n: A \vec{x} \le \vec{b}\}.$$
\end{definition}

Naturally, $\mathbb{Z}^n$ represents the set of all integer points (points with all integer coordinates).
Thus integer models of the linear constraints can be represented by $\{\vec{x} \in \mathbb{Z}^n : A \vec{x} \le \vec{b}\}$.
It is the same as the integer points inside the corresponding polytope, i.e.,
$$\{\vec{x} \in \mathbb{Z}^n : A \vec{x} \le \vec{b}\} = P \cap \mathbb{Z}^n.$$
In this paper, we assume that the polytopes are bounded, i.e., finite number of integer solutions,
otherwise, it can be easily detected via Integer Linear Programming (ILP).
Note that in our experiments, the running time of ILP is usually negligible compared to that of the integer counting.

\begin{definition}
More notations:
\begin{itemize}
\item Let $A = (\vec{A_1}, ..., \vec{A_m})^T$ and $\mathbf{h_i} = \vec{A_i}\vec{x} \le b_i$,
given $P = \{A\vec{x} \le \vec{b}\}$, i.e., $P = \mathbf{h_1} \cap ... \cap \mathbf{h_m}$.
\item Let $\text{Vol}(K)$ denote the volume of a given convex set $K$, which is the Lebesgue measure of $|K|$ in Euclidian space.
\item Let $C(\vec{x})$ denote the unit cube centered at $\vec{x}$.
\item Let $B(\vec{x}, r)$ denote the ball centered at $\vec{x}$, of radius $r$.
\end{itemize}
\end{definition}

\subsection{Hit-and-run Method}\label{sect:hr}

Hit-and-run random walk method was first introduced in \cite{BerbeeBKSST87},
whose limiting distribution is proved to be uniform.
It was employed and improved for volume approximation by \cite{Lovasz99,LovaszV06b}.
Experiments\,\cite{GeMZZ18} showed that a variation called Coordinate Directions Hit-and-run is more efficient in practice.
Thus we also adopt this variation, which is called Hit-and-run for short in the rest of paper.
It samples a real point from $\vec{p}$ in a given convex body $K$ by the following steps:
\begin{itemize}
\item Select a direction from $n$ coordinates uniformly.
\item Generate the line $l$ through $\vec{p}$ with above direction.
\item Pick a next point $\vec{p'}$ uniformly from $l \cap K$.
\item Start from $p'$ and repeat above steps $w$ times.
\end{itemize}
Earlier works\,\cite{LovaszV06b} proved that Hit-and-run method mixes in $w = O(n^2)$ steps for a random initial point and $O(n^3)$ steps for a fixed initial point.
However, further numerical studies\,\cite{LovaszD12,GeMZZ18} reported that $w = n$ is sufficient for nearly uniformly sampling in polytopes with dozens of dimensions.

\subsection{Multiphase Monte-Carlo Algorithm}\label{sect:mmc}

Multiphase Monte-Carlo Algorithm (MMC) is a polynomial time randomized algorithm, which was first introduced in \cite{DyerFK91}.
At first, the complexity is $O^*(n^{23})$\footnote{The ``soft-O'' notation $O^*$ indicates that
factors of $\log n$ and factors depending on other parameters like $\epsilon$ are suppressed.},
it was reduced to $O^*(n^3)$ by a series of works\,\cite{Lovasz99,LovaszV06a, CousinsV18}.
It consists of the following steps:
\begin{itemize}
\item Employ an Ellipsoid method to obtain an affine transformation $T$, s.t.,
    $B(\vec{0}, 1) \subset T(P) \subset B(\vec{0}, \rho)$, given a $\rho > n$.
    Note that $\text{Vol}(P) = \text{Vol}(T(P)) \cdot \det(T)$.
\item Construct a series of convex bodies $K_i = T(P) \cap B(\vec{0}, 2^{i/n})$, $i = 0, ..., l$,
    where $l = \lceil n \log_2 \rho \rceil$.
    Then $$\text{Vol}(T(P)) = \text{Vol}(K_l) = \text{Vol}(K_0) \cdot \prod_{i=0}^{l-1} \frac{\text{Vol}(K_{i+1})}{\text{Vol}(K_i)}.$$
    Specifically, $K_0 = B(\vec{0}, 1)$ and $K_l = T(P)$.
\item Generate a set $S_i$ of sample points by Hit-and-run in $K_{i+1}$, where $|S_i| = f(l, \epsilon, \delta)$.
    Then count $|K_i \cap S_i|$ and use $r_i = \frac{|K_i \cap S_i|}{|S_i|}$ to approximate the ratio $\frac{\text{Vol}(K_{i+1})}{\text{Vol}(K_i)}$.
\item At last, $\text{Vol}(P) \approx \text{Vol}(B(\vec{0},1)) \cdot \prod_{i=0}^{l-1} r_i \cdot \det(T)$.
\end{itemize}

Note that the function $f(l, \epsilon, \delta)$ determines the number of samples with given $\epsilon$, $\delta$,
s.t., relative errors of outputs are bounded in $[1 - \epsilon, 1 + \epsilon]$ with probability at least $1 - \delta$.

\section{Algorithm}\label{sect:alg}

To apply MMC framework and Hit-and-run random walk on lattice counting problem,
there are some difficulties:
\begin{itemize}
\item How to efficiently sampling lattice points nearly uniformly inside a polytope?
\item How to construct a chain of polytopes and then approximate ratios among them like MMC?
\item How many sample points are sufficient, given $\epsilon$, $\delta$? Could relative errors be computed while algorithm running?
\end{itemize}
In this section, we will propose new algorithms to answer above questions, with theoretical analysis.

\subsection{Lattice Sampling}

To sampling lattice points in a given polytope $P$,
we apply Hit-and-run random walk method with rejection sampling.

Intuitively, a real point $\vec{p} = (p_1, ..., p_n)$ corresponds to a lattice point $\vec{[p]} = ([p_1], ..., [p_n])$.
So lattice points can be generated by Hit-and-run method and number rounding, noted $[\,.\,]$.
However, the distribution of lattices generated by sampling real points directly in $P$ is not uniform.
Because the probability of sampling a lattice point $\vec{u}$ closed to polytopes' facets may be smaller than a point $\vec{v}$ which $C(\vec{v}) \subset P$.

\begin{example}
In Figure\,\ref{fig:shift}, the probability of a blue point picked by sampling directly in $P$, is smaller than a red point.
Now let us consider shifting $\mathbf{c}$ to $\mathbf{l_1}$, $\mathbf{l_2}$ and $\mathbf{l_3}$.
Note that $C(u_3) \subset \mathbf{a} \cap \mathbf{b} \cap \mathbf{l_2} \subset \mathbf{a} \cap \mathbf{b} \cap \mathbf{l_3}$,
but $C(u_3) \not\subset \mathbf{a} \cap \mathbf{b} \cap \mathbf{l_1}$.
Then the probability of picking $u_3$ by sampling real points in $\mathbf{a} \cap \mathbf{b} \cap \mathbf{l_2}$ or $\mathbf{a} \cap \mathbf{b} \cap \mathbf{l_3}$ is the same as red points.
\end{example}

\begin{figure}[htb]
	\centering
	\includegraphics[width=0.28\textwidth]{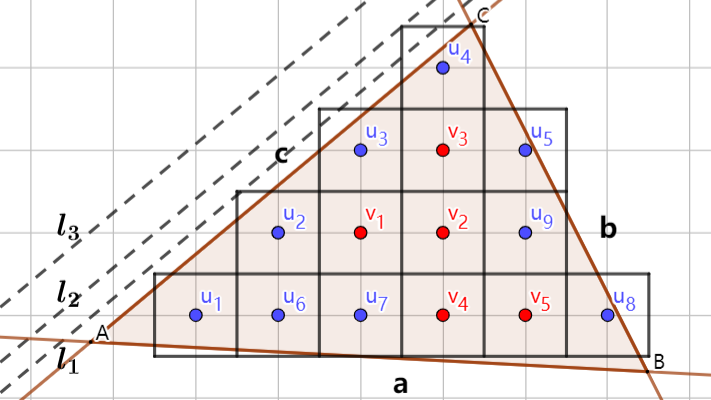}
	\caption{An illustration of shifting facets. Here $P = \triangle ABC = \mathbf{a} \cap \mathbf{b} \cap \mathbf{c}$,
        where $\mathbf{a}$, $\mathbf{b}$ and $\mathbf{c}$ are inequalities (half-spaces) correspond to $AB$, $BC$, $AC$, respectively.
        Inequalities $\mathbf{l_1}, \mathbf{l_2}, \mathbf{l_3}$ are parallel to $\mathbf{c}$.
        Red points $\{v_1,...,v_5\}$ and blue points $\{u_1,...,u_9\}$ are lattice points in $P$ s.t.
        $C(v_i) \subset P$ and $C(u_i) \not\subset P$ respectively.
    }\label{fig:shift}
\end{figure}

Therefore, our approach first enlarges $P$ to $P'$ by shifting facets of $P$.
Then it repeatedly generates real points in $P'$ and rejects samples whose corresponding lattice points outside $P$.
Obviously, the larger $P'$,  the larger probability of rejection.
Now we have a further question:
\begin{itemize}
\item How to obtain such $P'$ as small as possible?
\end{itemize}

Naturally, $P'$ should contain all unit cubes centered at lattice points in $P$, i.e.,
$$C(\vec{p}) \subset P', \forall \vec{p} \in P \cap \mathbb{Z}^n.$$
Without loss of generality, let us consider shifting the $i$th facet $\vec{A_i}\vec{x} \le b_i$.
The hyperplane shifting problem is equivalent to the following optimization problem
\begin{align*}
& \min b_i'\ \ \mathrm{s.t.}\ C(\vec{p}) \subset \vec{A_i}\vec{x} \le b_i', \forall \vec{p} \in P \cap \mathbb{Z}^n. \\
\Leftrightarrow & \max b_i'\ \ \mathrm{s.t.}\ [C(\vec{p}) \cap \vec{A_i}\vec{x} = b_i'] \ne \emptyset, \exists \vec{p} \in P \cap \mathbb{Z}^n. \\
\Leftrightarrow & \max \vec{A_i}\vec{x}\ \ \mathrm{s.t.}\ \vec{x} \in \bigcup_{\vec{p} \in P \cap \mathbb{Z}^n} C(\vec{p}), \vec{x} \in \mathbb{R}^n.
\end{align*}
In the worst case, assume there is a lattice point $\vec{q}$ on the $i$th facet of $P$, i.e., $\vec{A_i}\vec{q} = b_i$.
Then we have
\begin{align}
\Leftrightarrow & \max \vec{A_i}\vec{x}\ \ \mathrm{s.t.}\ \vec{x} \in C(\vec{q}), \vec{x} \in \mathbb{R}^n. \nonumber \\
\Leftrightarrow &\ \ b_i + \max \vec{A_i}\vec{x}\ \ \mathrm{s.t.}\ \vec{x} \in C(\vec{0}), \vec{x} \in \mathbb{R}^n. \label{eqn:shiftopt}
\end{align}
The optimization problem of Equation\,(\ref{eqn:shiftopt}) can be solved by Linear Programming (LP), e.g., Simplex algorithm.

\begin{algorithm}[htb]
\caption{Sample() -- Sample $s$ lattice points in $P$}
\label{alg:sample}
\textbf{Input}: $P$, $s$ \\
\textbf{Parameter}: $w$ \\
\textbf{Output}: $S$
\begin{algorithmic}[1] 
\FOR {\textbf{each} $\vec{A_i} \vec{x} \le b_i$ in $P$}
    \STATE $v_i \leftarrow$ Simplex($\max \vec{A_i}\vec{x}\ \ \mathrm{s.t.}\ \{-\frac{1}{2} \le x_i \le \frac{1}{2}\}$)
\ENDFOR
\STATE $P' \leftarrow$ $\{A\vec{x} \le \vec{b} + \vec{\mathit{v}}\}$
\STATE $T \leftarrow$ Ellipsoid($P'$)
\STATE $\vec{p} \leftarrow \vec{0}$, $S \leftarrow \emptyset$
\STATE {\textbf{repeat} $s$ \textbf{times}}
\begin{ALC@g}
    \STATE \textbf{do}
    \begin{ALC@g}
        \STATE $\vec{p} \leftarrow$ HitAndRun($T(P')$, $\vec{p}$, $w$)
        \STATE $\vec{q} \leftarrow [T^{-1}(\vec{p})]$
    \end{ALC@g}
    \STATE \textbf{while} $\vec{q} \not\in P$
    \STATE $S \leftarrow S \cup \{\vec{q}\}$
\end{ALC@g}
\STATE \textbf{end repeat}
\end{algorithmic}
\end{algorithm}

Algorithm\,\ref{alg:sample} is the pseudocode of our sampling method.
It first enlarges $P$ to $P'$ by the shifting method.
Next it applies the Shallow-$\beta$-Cut Ellipsoid method on $P'$ which is the same as MMC.
It obtains an affine transformation $T$ such that $B(0, 1) \subset T(P') \subset B(0, 2n)$.
Then it samples a lattice point $\vec{q}$ by $[T^{-1}(\vec{p})]$,
where $\vec{p}$ is a real sample point generated by Hit-and-run in $T(P')$, and $T^{-1}$ is the inverse transformation of $T$.
The algorithm only accepts samples inside $P$.
At last it repeats above steps till $|S| = s$.
The parameter $w$ will be discussed later in Section\,\ref{sect:imple}.

Why we adopt an affine transformation $T$ before random walks?
Intuitively, it could transform a `thin' polytope $P'$ into is a well-rounded one $T(P')$.
Thus it is easier for Hit-and-run walks to get out of corners.

The following results show that Algorithm\,\ref{alg:sample} generates lattice sample points in nearly uniform.
Recall that $\text{Vol}(P) = \text{Vol}(T(P)) \cdot \det(T)$.

\begin{lemma}\label{lem:accept}
The probability of acceptance is $\frac{|P \cap \mathbb{Z}^n|}{\text{Vol}(P')}$,
if \textnormal{Hit-and-run} is a uniform sampler.
\end{lemma}
\begin{proof}
Assume $\vec{x}$ is generated by Hit-and-run method. Then
\begin{align*}
& \text{Prob}(\vec{x}\ \mathrm{accepted}) = \text{Prob}([T^{-1}(\vec{x}))] \in P) \\
= & \text{Prob}(\vec{x} \in \cup_{\vec{p} \in P \cap \mathbb{Z}^n} T(C(\vec{p}))) = \frac{\text{Vol}(\cup_{\vec{p} \in P \cap \mathbb{Z}^n} T(C(\vec{p})))}{\text{Vol}(T(P'))} \\
= & \frac{\sum_{\vec{p} \in P \cap \mathbb{Z}^n} \text{Vol}(C(\vec{p})) / \det(T)}{\text{Vol}(P') / \det(T)} = \frac{|P \cap \mathbb{Z}^n|}{\text{Vol}(P')}. \qedhere
\end{align*}
\end{proof}

\begin{theorem}\label{lem:uniform}
Each point $\vec{x} \in P \cap \mathbb{Z}^n$ gets picked with the same probability,
if \textnormal{Hit-and-run} is a uniform sampler.
\end{theorem}
\begin{proof}
Consider an arbitrary point $\vec{x} \in P \cap \mathbb{Z}^n$.
Let $\vec{p}$ represents a real point generated by Hit-and-run in $T(P')$.
Then
\begin{align*}
& \text{Prob}(\vec{x}\ \mathrm{picked}) = \frac{\text{Prob}(\vec{p} \in T(C(\vec{x})))}{\text{Prob}(\vec{p}\ \mathrm{accepted})} \\
& \quad\quad = \frac{\text{Vol}(T(C(\vec{x})))}{\text{Vol}(T(P'))} \cdot \frac{\text{Vol}(P')}{|P \cap \mathbb{Z}^n|} = \frac{1}{|P \cap \mathbb{Z}^n|}. \qedhere
\end{align*}
\end{proof}

From Lemma\,\ref{lem:accept}, we observe that the acceptance could be very small when $|P \cap \mathbb{Z}^n| \ll \text{Vol}(P')$.

\subsection{Polytopes Chain Generation}

Now we consider a chain of polytopes $\{P_0, ..., P_l\}$ s.t.
\begin{equation}\label{eqn:subdiv}
\begin{split}
& |P \cap \mathbb{Z}^n| = |P_0 \cap \mathbb{Z}^n| \cdot \prod_{i=0}^{l-1} \frac{|P_{i+1} \cap \mathbb{Z}^n|}{|P_i \cap \mathbb{Z}^n|}, \\
& \frac{|P_{i+1} \cap \mathbb{Z}^n|}{|P_i \cap \mathbb{Z}^n|} \in
\begin{cases}
[r_{min}, r_{max}] & i \le l-2,\\
[r_{min}, 1) & i = l - 1.
\end{cases}
\end{split}
\end{equation}
where $[r_{min}, r_{max}]$ bounds ratios close to $\frac{1}{2}$, like $[0.4, 0.6]$.
If ratios are close to $0$, the computational cost of generating points in $P_{i+1} \cap \mathbb{Z}^n$ by sampling in $P_i \cap \mathbb{Z}^n$ will increase.
On the other hand, $l$ will be a large number when ratios are close to $1$,
which is also not computational-wise.
Algorithm~\ref{alg:subdiv} presents our method for constructing such $P_i$s.

\begin{algorithm}[htb]
\caption{Subdivision() -- Obtain the polytopes chain}
\label{alg:subdiv}
\textbf{Input}: $P$, $\mathit{s}$ \\
\textbf{Parameter}: $r_{max}$, $r_{min}$, $\mu$ \\
\textbf{Output}: $l$, $\{P_i\}$, $\{S_i\}$
\begin{algorithmic}[1] 
\STATE $P_0 \leftarrow$ GetRect($P$)
\STATE $i \leftarrow 0$, $j \leftarrow 1$ and $S_0 \leftarrow \emptyset$
\WHILE {$j \le m$}
	\STATE $S_i \leftarrow$ Sample($P_i$, $\mathit{s}$)
    \STATE $H \leftarrow \mathbb{R}^n$
    \WHILE {$\frac{|S_i \cap H \cap \mathbf{h_j}|}{|S_i|} > r_{max}$ and $j \le m$}
        \STATE $H \leftarrow H \cap \mathbf{h_j}$, $j \leftarrow j + 1$
    \ENDWHILE
    \STATE $k \leftarrow \min(j, m)$
    \IF {$\frac{|S_i \cap H \cap \mathbf{h_k}|}{|S_i|} \ge r_{min}$}
        \STATE $H \leftarrow H \cap \mathbf{h_k}$, $j \leftarrow k + 1$
    \ELSE
        \STATE \textbf{do}\label{line:bjbeg}
        \begin{ALC@g}
            \STATE $\vec{A_k'} \leftarrow \vec{A_k}$ or Disturb($A_k$, $\mu$) since second loop
            \STATE find min $b_k'$ s.t. $\frac{|S_i \cap H \cap \vec{A_k'} \vec{x} \le b_k'|}{|S_i|} \ge r_{min}$, $b_k' \ge b_k$
        \end{ALC@g}
        \STATE \textbf{while} no feasible $b_k'$ found \label{line:bjend}
        \STATE $H \leftarrow H \cap \vec{A_k'} \vec{x} \le b_k'$
    \ENDIF
    \STATE $P_{i+1} \leftarrow$ $P_i \cap H$
	\STATE $i \leftarrow i + 1$
\ENDWHILE
\STATE \textbf{return} $i$, $\{P_0, ..., P_i\}$, $\{S_0, ..., S_i\}$
\end{algorithmic}
\end{algorithm}

Recall that in MMC, it eventually approximates the ratio between volume of $P$ and an inner ball $B(\vec{0}, 1)$ whose exact volume is easy to compute.
It constructs a series of convex body $K_i$ inside $P$.
Lemma\,\ref{lem:accept} indicates that the smaller polytope, the more difficult to sampling lattice points,
naturally, we would like to construct polytopes chain outside $P$.
Our approach starts from an $n$-dimensional rectangle $P_0 = Rect(P) \supset P$, whose exact lattice count is also easy to obtain.
Next it constructs $P_1 \supset P$ by adding new cutting constraints on $P_0$, s.t. $\frac{|P_1 \cap \mathbb{Z}^n|}{|P_0 \cap \mathbb{Z}^n|}$ close to $\frac{1}{2}$.
Then it repeatedly generates $P_1 \supset P_2 \supset ...$ until a polytope $P_l = P$ found.

\begin{itemize}
\item How to find cutting constraints to halve $P_i$s?
\end{itemize}

\begin{example}
In Figure\,\ref{fig:subdiv}, given $P = \triangle ABC = \mathbf{a} \cap \mathbf{b} \cap \mathbf{c}$,
and $P_0 = ADEF \supset P$.
Now we try to cut $P_0$ with $\mathbf{a}$, $\mathbf{b}$ and $\mathbf{c}$.
We observe that $\frac{|P_0 \cap \mathbf{a} \cap \mathbf{b} \cap \mathbb{Z}^n|}{|P_0 \cap \mathbb{Z}^n|} = \frac{10}{15} > r_{max}$
and $\frac{|P_0 \cap \mathbf{a} \cap \mathbf{b} \cap \mathbf{c} \cap \mathbb{Z}^n|}{|P_0 \cap \mathbb{Z}^n|} = \frac{4}{15} < r_{min}$.
Then we find $\mathbf{d}$ parallel to $\mathbf{c}$ s.t. $\frac{|P_0 \cap \mathbf{a} \cap \mathbf{b} \cap \mathbf{d} \cap \mathbb{Z}^n|}{|P_0 \cap \mathbb{Z}^n|} = \frac{8}{15}$.
Thus $P_1 = P_0 \cap \mathbf{a} \cap \mathbf{b} \cap \mathbf{d}$.
\end{example}

\begin{figure}[htb]
	\centering
	\includegraphics[width=0.26\textwidth]{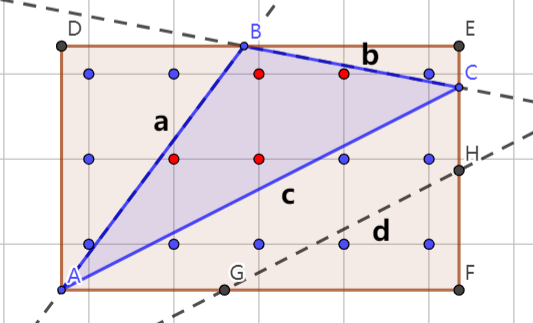}
	\caption{An example of constructing $P_0 = ADEF$, $P_1 = ABCHG$ and $P_2 = \triangle ABC = P$.}\label{fig:subdiv}
\end{figure}

Suppose that we already have $P_0 \supset ... \supset P_i$ which $P_i \subset \mathbf{h_1} \cap ... \cap \mathbf{h_{j-1}}$ and $P_i \not\subset \mathbf{h_j}$.
Then cutting constraints for constructing $P_{i+1}$ are found by the following steps:
\begin{itemize}
\item \textbf{Step 1.} Add constraints $\mathbf{h_j}$, $\mathbf{h_{j+1}}$,... repeatedly
until a $k$ is found s.t. $\frac{|P_i \cap \mathbf{h_j} \cap ... \cap \mathbf{h_k} \cap \mathbb{Z}^n|}{|P_i \cap \mathbb{Z}^n|} \le r_{max}$ or $k = m$.
\item \textbf{Step 2.} If $\frac{|P_i \cap \mathbf{h_j} \cap ... \cap \mathbf{h_k} \cap \mathbb{Z}^n|}{|P_i \cap \mathbb{Z}^n|} \ge r_{min}$,
then $P_{i+1} = P_i \cap \mathbf{h_j} \cap ... \cap \mathbf{h_k}$ has been found.
Note that $P_{i+1} = P_l= P$ when $k = m$.
\item \textbf{Step 3.} Otherwise, it indicates that $\mathbf{h_k}$ over-cuts the solution space.
Then we find an $\mathbf{h_k'} = \vec{A_k'}\vec{x} \le b_k'$ (almost) parallel to $\mathbf{h_k}$
s.t. $r_{min} \le \frac{|P_i \cap \mathbf{h_j} \cap ... \cap \mathbf{h_k'} \cap \mathbb{Z}^n|}{|P_i \cap \mathbb{Z}^n|} \le r_{max}$.
At last, let $P_{i+1} = P_i \cap \mathbf{h_j} \cap ... \cap \mathbf{h_{k-1}} \cap \mathbf{h_k'}$.
\end{itemize}
About above steps, we may naturally ask:

\begin{itemize}
\item How to determine the value of $\frac{|P_i \cap \mathbf{h_j} \cap ... \cap \mathbf{h_k} \cap \mathbb{Z}^n|}{|P_i \cap \mathbb{Z}^n|}$?
\end{itemize}

Algorithm\,\ref{alg:subdiv} samples lattice points $S_i$ in $P_i$
and then approximates $\frac{|P_i \cap \mathbf{h_j} \cap ... \cap \mathbf{h_k} \cap \mathbb{Z}^n|}{|P_i \cap \mathbb{Z}^n|}$ via $\frac{|S_i \cap \mathbf{h_j} \cap ... \cap \mathbf{h_k}|}{|S_i|}$.
Since we aim to obtain $P_{i+1}$ s.t. $\frac{|P_{i+1} \cap \mathbb{Z}^n|}{|P_i \cap \mathbb{Z}^n|}$ close to $\frac{1}{2}$,
it is not necessary to approximate very accurately with a mass of samples.

\begin{itemize}
\item How to find $\mathbf{h_k'}$ in Step 3?
\end{itemize}

Line\,\ref{line:bjbeg} to\,\ref{line:bjend} in Algorithm\,\ref{alg:subdiv} is the loop of finding $\mathbf{h_k'}$.
At the first time of loop, it sets $\vec{A_k'} = \vec{A_j}$ and searches the minimum $b_k' \ge b_k$ s.t. $\frac{|S_i \cap H \cap \vec{A_k'} \vec{x} \le b_k'|}{|S_i|} \ge r_{min}$.
We then compute and sort $D = \{d : d = \vec{A_k}\vec{p},\ \forall \vec{p}\in S_i \cap H\}$.
Thus searching $b_k'$ is equivalent to scanning $D$,
whose time complexity is $O(|D|) = O(|S_i \cap H|) = O(s)$.

Note that there may be no feasible $b_k'$,
as for certain $y$,
$\frac{|S_i \cap H \cap \vec{A_k} \vec{x} \le y|}{|S_i|} > r_{max}$,
$\frac{|S_i \cap H \cap \vec{A_k} \vec{x} <y|}{|S_i|} < r_{min}$.
For example, $|x_1 + x_2 = 0.99 \cap \mathbb{Z}^2| = 0$ and $|x_1 + x_2 = 1 \cap \mathbb{Z}^2| = \infty$.
Therefore, if our algorithm fails to find a feasible $b_k'$ once,
it will generate $\vec{A_k'} = \{a_{k1}', ..., a_{kn}'\}$ by disturbing $\vec{A_k}$,
i.e., $a_{ki}' \in [a_{ki} - \mu, a_{ki} + \mu]$,
where $\mu \in \mathbb{R}$ is a small constant.
In practice, the loop in line 13--16 (Algorithm 2) usually finds a feasible $b_k'$ by disturbing $\vec{A_k'}$ once, occasionally twice, though the loop may not stop in theory in worst cases.

With respect to the size of $l$, it is easy to find the following result as every $P_{i+1}$ is constructed by nearly halve $P_i$.

\begin{theorem}
The length $l$ of the chain $P_0, .., P_l$ constructed by Algorithm\,\ref{alg:subdiv} is in $O(\log_2 |P_0 \cap \mathbb{Z}^n|)$ in the worst case.
\end{theorem}

\subsection{Dynamic Stopping Criterion}

Approximating $|P \cap \mathbb{Z}^n|$ is factorized into approximating a series ratios $\frac{|P_{i+1} \cap \mathbb{Z}^n|}{|P_i \cap \mathbb{Z}^n|}$ by Equation\,(\ref{eqn:subdiv}).
Naturally, we could approximate ratios via $\frac{|P_{i+1} \cap S_i|}{|S_i|}$,
where $S_i$ is a set of lattice points sampled in $P_i$ by Algorithm\,\ref{alg:sample}.
A key question rises:
\begin{itemize}
\item How many sample points is sufficient to approximate $|P \cap \mathbb{Z}^n|$ with certain guarantees, like an $(\epsilon, \delta)$-bound?
\end{itemize}

Let $R_i$ denote the random variable of $\frac{|P_{i+1} \cap S_i|}{|S_i|}$,
and $R = \prod_{i=0}^{l-1}{R_i}$.
Note that $R_i$s are mutually independent, since for each $S_i$, the random walk starts from origin $\vec{0} \in T(P')$ (see Algorithm\,\ref{alg:sample}).
Thus we have the variance of $R$:
\begin{align}
\mathrm{Var}(R) & = \mathrm{Var}(\prod{R_i}) = \mathrm{E}((\prod{R_i})^2) - [\mathrm{E}(\prod{R_i})]^2 \nonumber \\
& = \prod{\mathrm{E}(R_i^2)} - [\prod\mathrm{E}(R_i)]^2 \nonumber \\
& = \prod{[\mathrm{Var}(R_i) + \mathrm{E}(R_i)^2]} - \prod{[\mathrm{E}(R_i)]^2}. \label{eqn:varr}
\end{align}
From Chebyshev inequality, we have:
\begin{align}
& \text{Prob}(|\frac{R - \mathrm{E}(R)}{E(R)}| \ge \epsilon) \le \frac{\mathrm{Var}(R)}{\epsilon^2 \cdot E(R)^2} \le \delta \nonumber \\
& \Rightarrow \mathrm{Var}(R) \le \delta \cdot \epsilon^2 \cdot E(R)^2. \label{eqn:bound}
\end{align}
Equation\,(\ref{eqn:bound}) shows when the approximate result lies in $[(1 - \epsilon)|P \cap \mathbb{Z}^n|, (1 + \epsilon)|P \cap \mathbb{Z}^n|]$ with probability at least $1 - \delta$,
i.e., satisfies an $(\epsilon, \delta)$-bound.
Thus we adopt Equation\,(\ref{eqn:bound}) as the stopping criterion of approximation.

Given a set of sample $S_i$, let $r_i = |P_{i+1} \cap S_i| / |S_i|$ and $r = \prod_{i=0}^{l-1}r_i$.
We use $r$ and $r_i$s to approximately represent $E(R_i)$s and $E(R)$ respectively (Lemma\,\ref{lem:bound2} shows that such substitutions are safe).
Then we split $S_i$ into $N$ groups $\{S_{ij}\}$ with the same size $s / \gamma$,
where $N = |S_i| \cdot \gamma / s$ is the number of groups.
Let $r_{ij} = |P_{i+1} \cap S_{ij}| / |S_{ij}|$ and $R_{ij}$ denote the random variable of $r_{ij}$.
If $R_{ij}$s are mutually independent and follow the same distribution, we have
\begin{align*}
\mathrm{Var}(R_i) & = \mathrm{Var}(\frac{\sum_{j=1}^N R_{ij}}{N}) = \frac{1}{N^2} \sum_{j=1}^N \mathrm{Var}(R_{ij}) \\
& = \frac{\mathrm{Var}(R_{i1})}{N} \approx \frac{1}{N} \sum_{j=1}^N \frac{(r_{ij} - r_i)^2}{N - 1}.
\end{align*}
Note that $R_{ij}$ can be exactly mutually independent if random walks start from a fixed point,
however, it is not actually necessary.
Let $v_i = \sum \frac{(r_{ij} - r_i)^2}{N(N - 1)}$.
As a result, an approximate stopping criterion is obtained
\begin{equation}
\mathrm{Var}(R) \approx \prod (v_i + r_i^2) - r^2 \le \delta \cdot \epsilon^2 \cdot r^2. \label{eqn:bound_approx}
\end{equation}

\begin{algorithm}[htb]
\caption{Approximate Lattice Counts}
\label{alg:alc}
\textbf{Input}: $P$ \\
\textbf{Parameter}: $\epsilon$, $\delta$, $s$, $\gamma$ \\
\textbf{Output}: $|P \cap \mathbb{Z}^n|$
\begin{algorithmic}[1] 
\STATE $(l, \{P_i\}, \{S_i\}) \leftarrow$ Subdivision($P$, $\mathit{s}$)
\STATE $N \leftarrow 0$
\STATE \textbf{do}
\begin{ALC@g}
    \STATE $N \leftarrow N + \gamma$
    \FOR {$i = 0$ to $l-1$}
        \STATE $S_i \leftarrow S_i \cup$ Sample($P_i$, $\mathit{s}$)
        \STATE $r_i \leftarrow |P_{i+1} \cap S_i| / |S_{i}|$
        \STATE Split $S_i$ into $N$ groups $S_{i1}, ..., S_{iN}$
        \STATE $r_{ij} \leftarrow |P_{i+1} \cap S_{ij}| / |S_{ij}|$, \quad $j \in \{1, ..., N\}$
        \STATE $v_i \leftarrow \sum_{j=1}^N \frac{(r_{ij} - r_i)^2}{N(N - 1)}$
    \ENDFOR
    \STATE $r \leftarrow \prod_{i=0}^{l-1}{r_i}$
    \STATE $v \leftarrow \prod_{i=0}^{l-1}{(v_i + r_i)^2} - r^2$
\end{ALC@g}
\STATE \textbf{while} $v \le \delta \cdot \epsilon^2 \cdot r^2$
\STATE \textbf{return} $|P_0 \cap \mathbb{Z}^n| \cdot r$
\end{algorithmic}
\end{algorithm}

The pseudocode of the main framework is presented as Algorithm\,\ref{alg:alc}.
It first generates $s$ sample points for each $P_i$ and then computes $r_i$s, $v_i$s, $r$ and $v$.
If Equation\,(\ref{eqn:bound_approx}) satisfies, it returns $|P_0 \cap \mathbb{Z}^n| \cdot r$,
otherwise, it repeats above steps.

\begin{figure*}[htb]
    \centering
    \begin{minipage}{0.32\linewidth}
        \centerline{\includegraphics[width=\textwidth]{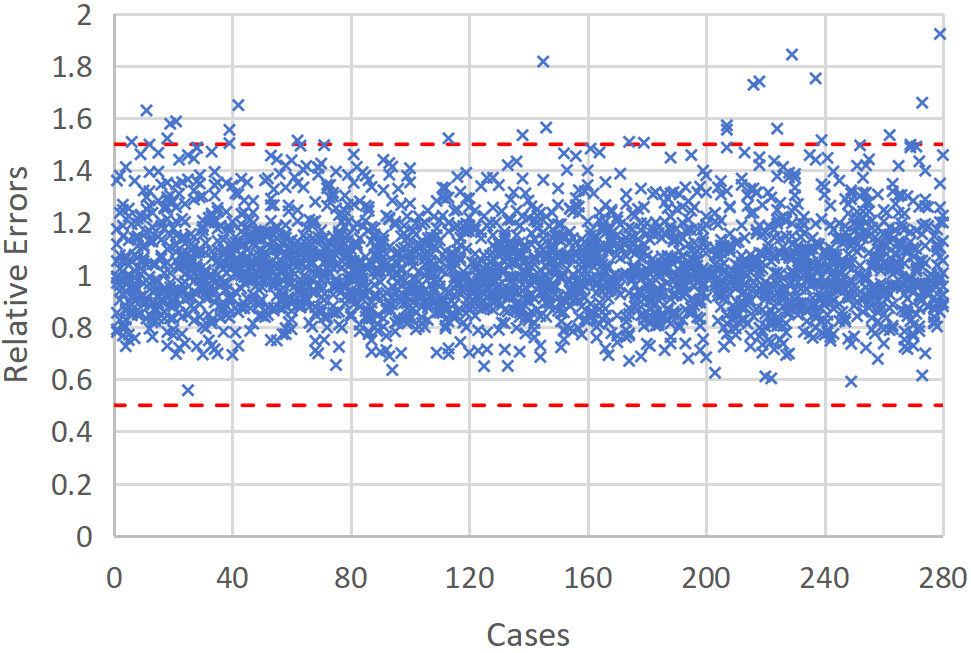}}
        \centerline{ (a) $\epsilon = 0.5$, $\delta = 0.1$}
    \end{minipage}
    \hspace{5pt}
    \begin{minipage}{0.32\linewidth}
        \centerline{\includegraphics[width=\textwidth]{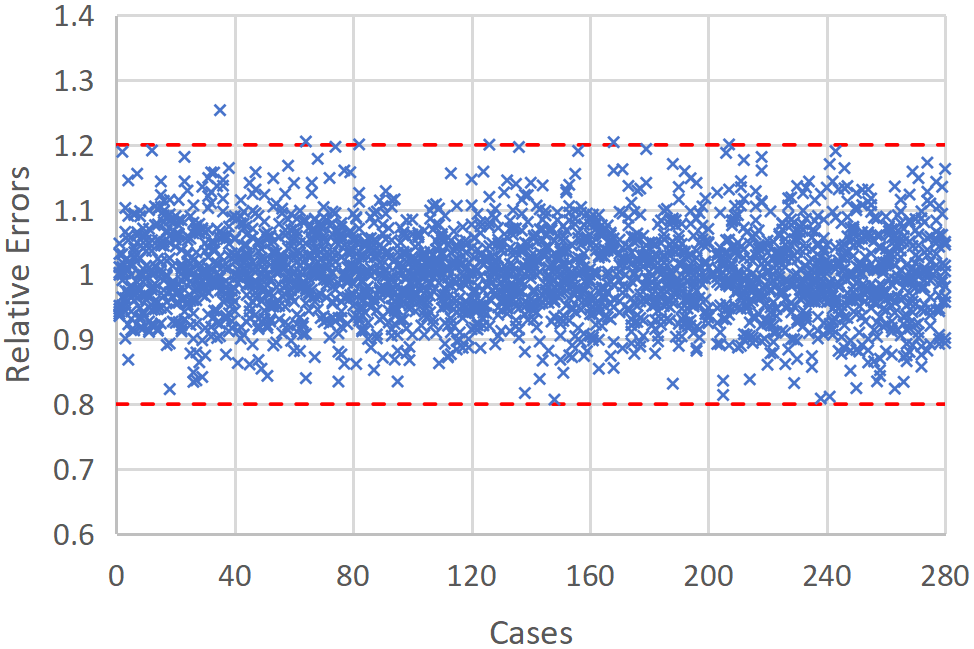}}
        \centerline{ (b) $\epsilon = 0.2$, $\delta = 0.1$}
    \end{minipage}
    \hspace{5pt}
    \begin{minipage}{0.32\linewidth}
        \centerline{\includegraphics[width=\textwidth]{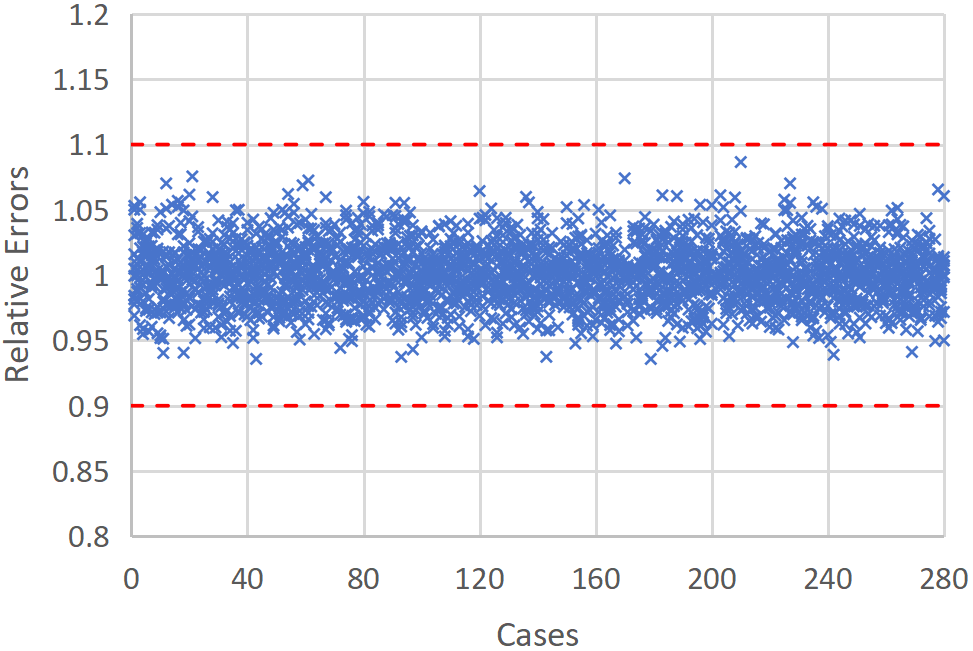}}
        \centerline{ (c) $\epsilon = 0.1$, $\delta = 0.05$}
    \end{minipage}
    \caption{
        Quality of counts computed by \textsc{ALC} with different $\epsilon$ and $\delta$
        on cases whose exact counts are available.
        Each case was experimented $10$ times, i.e., $10$ data points per case.
        The average running times in (a) (b) (c) are 0.19s, 0.81s, 5.73s respectively.
    }\label{fig:bound}
\end{figure*}

\begin{figure*}[htb]
\centering
    \begin{minipage}{0.355\linewidth}
        \centerline{\includegraphics[width=\textwidth]{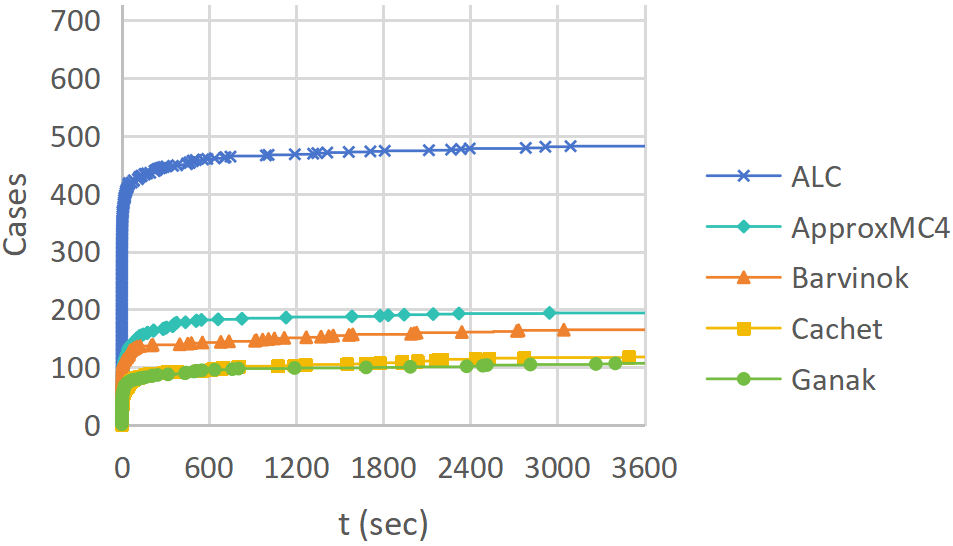}}
        \centerline{ (a) Random polytopes.\quad}
    \end{minipage}
    \hspace{2pt}
    \begin{minipage}{0.255\linewidth}
        \centerline{\includegraphics[width=\textwidth]{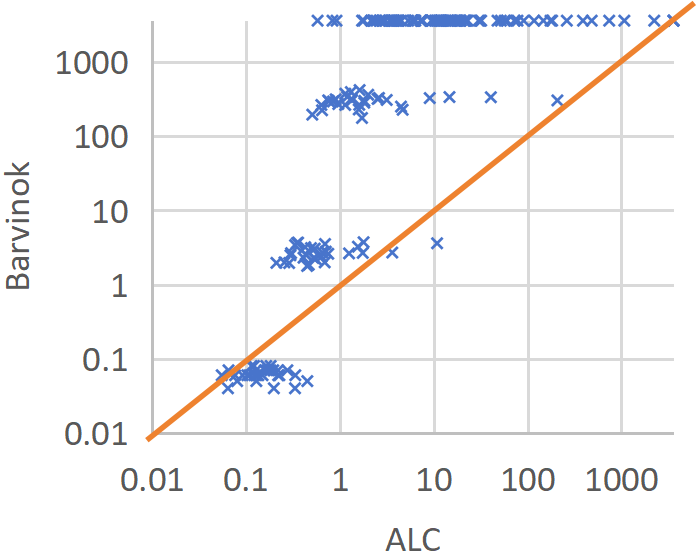}}
        \centerline{ (b) Rotated thin rectangles.}
    \end{minipage}
    \hspace{2pt}
    \begin{minipage}{0.355\linewidth}
        \centerline{\includegraphics[width=\textwidth]{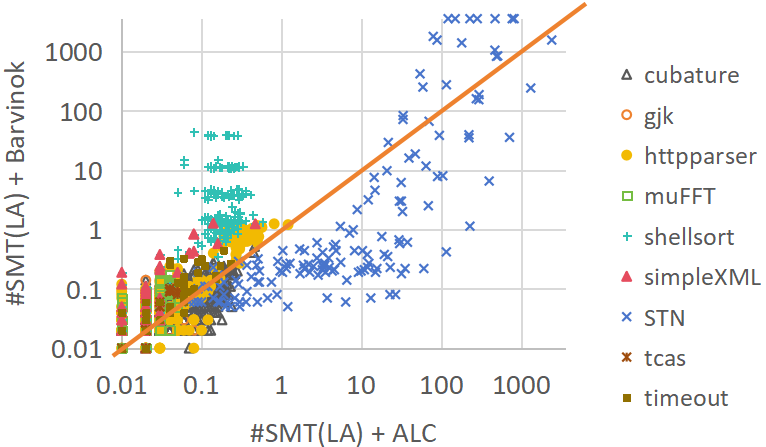}}
        \centerline{ (c) Application instances.}
    \end{minipage}
    \caption{Performance comparison among tools on different families of benchmarks.}\label{fig:comp}
\end{figure*}

\begin{lemma}
\label{lem:bound1}
$\lim_{|S_i|\rightarrow \infty} r_i = \frac{|P_{i+1} \cap \mathbb{Z}^n|}{|P_i \cap \mathbb{Z}^n|}$
and $\lim_{|S|\rightarrow \infty} r = \frac{|P \cap \mathbb{Z}^n|}{|P_0 \cap \mathbb{Z}^n|}$,
if Hit-and-run is a uniform sampler.
\end{lemma}
\begin{proof}
Note that sampling uniform in $P_i$ and then count the number of samples in $P_{i+1}$ is a Bernoulli trial.
\end{proof}

\begin{lemma}
\label{lem:bound2}
Equation\,(\ref{eqn:bound}) and\,(\ref{eqn:bound_approx}) are approximately equivalent,
regardless of the difference between $r$ and $E(R)$.
\end{lemma}
\begin{proof}
Let $c = r / E(R)$ and $c_i = r_i / E(R_i)$ represent the differences.
Since $r_i = c_i \cdot E(R_i) = E(c_iR_i)$, we have
$$v_i \approx \mathrm{Var}(c_iR_i) = c_i^2 \cdot \mathrm{Var}(R_i).$$
Equation\,(\ref{eqn:bound}) can be transformed into
\begin{align}
\delta \cdot \epsilon^2 & \ge \frac{\prod(\mathrm{Var}(R_i) + E(R_i)^2)}{E(R)^2} - 1 \nonumber \\
& \approx \frac{c^2}{E(R)^2} \cdot \prod \frac{\mathrm{Var}(R_i) + E(R_i)^2}{c_i^2} - 1 \nonumber \\
& = \frac{\prod (v_i + r_i^2)}{r^2} - 1. \label{eqn:bound_diff}
\end{align}
Note that Equation\,(\ref{eqn:bound_diff}) is the same as Equation\,(\ref{eqn:bound_approx}).
\end{proof}

From Lemma\,\ref{lem:bound1}, \ref{lem:bound2} and Equation\,(\ref{eqn:bound_approx}), we have
\begin{theorem}
\label{thm:bound2}
The output of Algorithm\,\ref{alg:alc} is approximately bounded in an $(\epsilon, \delta)$-bound.
\end{theorem}

\subsection{Implementation Details}\label{sect:imple}

The setting of parameters in Algorithm\,\ref{alg:sample},~\ref{alg:subdiv} and\,\ref{alg:alc} are listed with explanations as the following:
\begin{itemize}
\item {$\epsilon = 0.2$ and $\delta = 0.1$.}
    They determine the bounds of counts computed by our approach.
    Experimental results with more pairs of values, such as $(0.5, 0.1)$ and $(0.1, 0.05)$, can be found in Section\,\ref{sect:eval}.
\item {$w = n$.}
    It controls the number of Hit-and-run walks per real sample point.
    Earlier theoretical results\,\cite{LovaszV06b} showed the upper bounds on $w$ in the Markov chain is $O(n^2)$ for a random initial point and $O(n^3)$ for a fixed initial point.
    However, further numerical studies\,\cite{LovaszD12,GeMZZ18} reported that $w = n$ is sufficient on polytopes with dozens of dimensions.
    They also tried $w = 2n$ and $w = n \ln n$, but no visible improvement.
    Thus we adopt $w = n$.
\item {$s = 2 / (\delta \cdot \epsilon^2)$:}
    It controls the number of samples in one round.
    We select this value, as $1 / (\delta \cdot \epsilon^2)$ uniform samples are sufficient to approximate $r_i$ in $(\epsilon, \delta)$-bound.
    Note that the total number of samples is determined by the stopping criterion instead of $s$.
\item {$r_{min} = 0.4$, $r_{max} = 0.6$, $\mu = 0.005$ and $\gamma = 10$.}
\end{itemize}

In Algorithm\,\ref{alg:subdiv}, $P_0 = Rect(P)$ can be easily computed by LP or ILP.
Naturally, LP is cheaper than ILP, but the rectangle generated by ILP is smaller.
In practice, the cost of ILP is usually negligible compared to entire counting algorithm.

In Algorithm\,\ref{alg:subdiv} and\,\ref{alg:alc},
samples in $S_i \cap P_{i+1}$ can be reutilized in $S_{i+1}$.
Thus we only have to generate $s - |S_i \cap P_{i+1}|$ new samples for $S_{i+1}$.
\cite{GeMZZ18} proved that this technique has no side-effect on errors for approximating ratios.

\begin{table*}[htb]
\scriptsize\centering
\setlength\tabcolsep{4pt}
\begin{tabular}{|c|c|c|c|c|c|c|c|c|c|c|c|c|c|c|c|c|c|c|c|c|c|}
    \hline
    & Dim. $n$ & 3 & 4 & 5 & 6 & 7 & 8 & 9 & 10 & 11 & 12 & 13 & 14 & 15 & 20 & 30 & 40 & 50 & 60 & 70 & 80 \\
    \hline
    \multirow{3}{*}{ALC} & \#solved & \bf{33} & \bf{33} & \bf{33} & \bf{33} & \bf{33} & \bf{33} & \bf{33} & \bf{33} & \bf{33}
                        & \bf{30} & \bf{31} & \bf{29} & \bf{28} & \bf{22} & \bf{14} & \bf{11} & \bf{10} & \bf{5} & \bf{4} & \bf{1} \\
    & avg. $\bar{t}$ & \bf{0.03} & \bf{0.05} & \bf{0.07} & \bf{0.19} & \bf{0.65} & \bf{0.50} & \bf{2.42} & \bf{85.6} & \bf{55.5}
                        & \bf{42.6} & \bf{138} & \bf{48.8} & \bf{151} & \bf{156} & \bf{286} & \bf{249} & \bf{1057} & \bf{515} & \bf{1684} & \bf{3090} \\
    & avg. $\bar{l}$ & 1.6 & 3.2 & 4.2 & 6.1 & 7.4 & 8.2 & 9.5 & 12.3 & 13.6 & 13.4 & 14.4 & 15.8 & 17.6 & 21.6 & 24.7 & 31.6 & 40.4 & 31.4 & 40.3 & 44 \\
    \hline
    \multirow{2}{*}{Barvinok} & \#solved & \bf{33} & \bf{33} & \bf{33} & \bf{33} & 22 & 11 & 0 & 0 & 0 & 0 & 0 & 0 & 0 & 0 & 0 & 0 & 0 & 0 & 0 & 0 \\
    & avg. $\bar{t}$ & 0.22 & 0.24 & 2.72 & 105 & 1052 & 1158 & --- & --- & --- & --- & --- & --- & --- & --- & --- & --- & --- & --- & --- & --- \\
    \hline
    \multirow{2}{*}{Cachet} & \#solved
                     & 27   & 18    & 17    & 13    & 11    & 9     & 6     & 6     & 4    & 3    & 2    & 2    & 0 & 0 & 0 & 0 & 0 & 0 & 0 & 0 \\
    & avg. $\bar{t}$ & 161  & 71.8  & 537   & 396   & 291   & 434   & 153   & 601   & 735  & 483  & 621  & 2159 & --- & --- & --- & --- & --- & --- & --- & --- \\
    \hline
    \multirow{2}{*}{Ganak} & \#solved & 25 & 17 & 13 & 11 & 9 & 8 & 7 & 5 & 3 & 3 & 3 & 3 & 0 & 0 & 0 & 0 & 0 & 0 & 0 & 0 \\
    & avg. $\bar{t}$ & 68.7 & 256 & 9.4 & 187 & 27.6 & 198 & 459 & 169 & 59.2 & 390 & 1796 & 2909 & --- & --- & --- & --- & --- & --- & --- & --- \\
    \hline
    \multirow{2}{*}{ApproxMC4} & \#solved & \bf{33} & \bf{33} & 32 & 22 & 19 & 13 & 10 & 10 & 6 & 5 & 4 & 4 & 3 & 0 & 0 & 0 & 0 & 0 & 0 & 0 \\
    & avg. $\bar{t}$ & 1.16 & 9.78 & 81.3 & 98.4 & 136 & 121 & 360 & 580 & 287 & 608 & 352 & 673 & 1001 & --- & --- & --- & --- & --- & --- & --- \\
    \hline
\end{tabular}
\caption{More statistics of performance on random polytopes with respect to $n$ ($33$ cases for each $n$, experiment once per case).}\label{table:comp}
\end{table*}

\section{Evaluation}\label{sect:eval}

We implemented a prototype tool called \textsc{ApproxLatCount} (\textsc{ALC})
\footnote{Source code of \textsc{ALC} and experimental data including benchmarks can be found at https://github.com/bearben/ALC.} in C++.
Furthermore, we integrated \textsc{ALC} into a DPLL(T)-based \#SMT(LA) counter\,\cite{GeMZZ18}.
Experiments were conducted on Intel(R) Xeon(R) Gold $6248$ @ $2.50$GHz CPUs
with a time limit of 3600 seconds and memory limit of 4 GB per benchmark.
The setting of parameters of \textsc{ALC} has already been presented and discussed in Section\,\ref{sect:imple}.
The benchmark set consists of three parts:
\begin{itemize}
\item{Random Polytopes:}
We generated $726$ random polytopes with three parameters $(m, n, \lambda)$,
where $n$ ranges from $3$ to $100$, $m \in \{n/2, n, 2n\}$ and	$\lambda \in \{2^0, 2^1, ..., 2^{10}\}$.
A benchmark is in the form of $\{A\vec{x} \le \vec{b}, -\lambda \le x_i \le \lambda\}$,
where $a_{ij} \in [-10, 10] \cap \mathbb{Z}$ and $b_i \in [-\lambda, \lambda] \cap \mathbb{Z}$.
\item{Rotated Thin Rectangles:} To evaluate the quality of approximations on ``thin'' polytopes,
$180$ $n$-dimensional rectangles $\{-1000 \le x_1 \le 1000, -\tau \le x_i \le \tau, i \ge 2\}$ were generated
and then rotated randomly, where $n \in \{3, ..., 8\}$ and $\tau \in \{0.1, 0.2, ..., 2.9, 3.0\}$.
\item{Application Instances:}
We adopted $4131$ benchmarks \cite{GeB21} from program analysis and simple temporal planning.
The domain of variables is $[-32, 31]$.
\end{itemize}

We compared our tool \textsc{ALC} with the state-of-the-art integer counter \textsc{Barvinok}\,\cite{VerdoolaegeSBLB07}.
On random polytopes, we further compared our approach with the state-of-the-art propositional model counters
\textsc{ApproxMC4}\,\cite{SoosM19},
\textsc{Cachet}\,\cite{SangBBKP04},
and \textsc{Ganak}\,\cite{SharmaRSM19}.
We used the default settings of \textsc{ApproxMC4} ($\epsilon = 0.8$, $\delta = 0.2$) and \textsc{Ganak} ($\delta = 0.05$).
Note that they require CNF formulas as inputs.
Thus we first translated linear constraints into bit-vector formulas,
and then translated into propositional CNF with \textsc{Boolector}\,\cite{NiemetzPWB18}.
Translation times are not included in the running times.

Figures\,\ref{fig:bound}\,(a)\,(b)\,(c) show the relative errors (y-axis) of counts computed by \textsc{ALC} with different $(\epsilon, \delta)$ settings.
The experiments were conducted on random polytopes (case $91 \sim 280$) and rotated thin rectangles (case $1 \sim 90$) whose exact counts could be obtained by \textsc{Barvinok} and \textsc{Cachet}.
We run \textsc{ALC} $10$ times on each cases.
So there are $10$ data points per case, $2800$ data points per figure.
We observe that the counts computed by \textsc{ALC} are bounded well.
For example, in Figure\,\ref{fig:bound}\,(b), relative errors should lie in $[0.8, 1.2]$ with probability at least $90\%$ with $\epsilon = 0.2$, $\delta = 0.1$.

Figures\,\ref{fig:comp}\,(a)\,(b)\,(c) compare running times among tools on different families of benchmarks.
In general, \textsc{ALC} significantly outperforms other tools.
On random polytopes, more results with respect to $n$ are listed in Table\,\ref{table:comp}, which will be discussed later.
Figure\,\ref{fig:comp}\,(b) present the results on rotated thin rectangles.
Note that none of cases in this family was solved in timeout by \textsc{ApproxMC4}, \textsc{Cachet} or \textsc{Ganak},
due to larger coefficients and variable domains.
We observe jumps regarding running times of \textsc{Barvinok}, as $n$ increases.
Figure\,\ref{fig:comp}\,(c) compares on application instances
which are all SMT(LA) formulas.
Since we only integrated \textsc{ALC} and \textsc{Barvinok} into the \#SMT(LA) counter,
we did not compare with other tools.
Note that `STN' is the family of simple temporal planning benchmarks, others are all generated by analyzing C++ programs.
We find that most benchmarks were solved in one second by both tools, except `shellsort' and `STN'.
On `shellsort', \textsc{ALC} significantly outperforms \textsc{Barvinok}.
On `STN', \textsc{ALC} eventually gains upper hand as the dimensionality increases.

Table\,\ref{table:comp} lists the number of solved cases and average running times (exclude timeout cases) with respect to $n$.
For each $n$, there are $33$ benchmarks.
We find that \textsc{ALC} could handle random polytopes up to $80$ dimensions.
``Avg. $\bar{l}$'' means the average length of polytopes chain (exclude timeout cases), which grows nearly linear.
Note that \textsc{ApproxMC4}, \textsc{Cachet} and \textsc{Ganak} could solve cases with more variables (max to $15$) than \textsc{Barvinok} here,
due to benchmarks with $\lambda = 1$, i.e., $-1 \le x_i \le 1$, which are in favor of propositional model counters.

\section{Related Works}\label{sect:relate}

There are a few related works which also investigate approximate integer solution counting problem.
In \cite{KannanV97}, an algorithm for sampling lattice points in a polytope was introduced.
Similar to Algorithm\,\ref{alg:sample}, it considers an enlarged polytope $P''$ for real points sampling and then rejects samples outside $P$,
where $$P'' = \{\vec{x} : \vec{A_i}\vec{x} \le b_i + (c + \sqrt{2\log m})|\vec{A_i}|)\},$$
$c=\sqrt{\ln \frac{4}{\varepsilon}}$ and $\varepsilon$ is the variational difference between the uniform density and the probability density of real points sampling.
As a result, they proved that there exists a polynomial time algorithm for nearly uniform lattice sampling if $b_i \in \Omega(n\sqrt{m}|\vec{A_i}|)$.
However, in practice, such condition is often too loose.
For example, benchmarks considered in Section~\ref{sect:eval} are usually smaller, i.e., $b_i < n\sqrt{m}|\vec{A_i}|$, especially when $n \ge 10$,
which has a higher difficulty in sampling.
Also note that $P'$ computed by our approach is tighter than $P''$.
Thus the probability of rejection by sampling in $P'$ is lower than in $P''$.
In addition, back to the time of this work published, the best real points sampler is only with time complexity of $O^*(n^5)$.
Nowadays, the state-of-the-art real points sampler is in $O^*(n^3)$.

A more recent work \cite{GeB21} introduced factorization preprocessing techniques to reduce polytopes dimensionality.
Suppose a polytope $P$ has been factorized into $F_1, ..., F_k$, and $|P \cap \mathbb{Z}^n| = \prod_{i=1}^k |F_i \cap \mathbb{Z}^{n_i}|$,
where $n_i$ represents the dimensionality of $F_i$.
To approximate $|P \cap \mathbb{Z}^n|$ with given $\epsilon, \delta$,
we have to approximate counts in $F_i$ with smaller $\epsilon', \delta'$.
It indicates that factorization techniques integrated with \textsc{ALC} may not as effective as with exact counters.

\section{Conclusion}\label{sect:conclude}

In this paper, a new approximate lattice counting framework is introduced,
with a new lattice sampling method and dynamic stopping criterion.
Experimental results show that our algorithm significantly outperforms the state-of-the-art counters, with low errors.
Since our sampling method is limited by the Hit-and-run random walk,
which is only a nearly uniform sampler,
we are interested in an efficient method to test the uniformity of samplers in the future.

\section*{Acknowledgments}

This work is supported by National Key R\&D Program of China (2022ZD0116600).
Cunjing Ge is supported by the National Natural Science Foundation of China (62202218),
and is sponsored by CCF-Huawei Populus Grove Fund (CCF-HuaweiFM202309).

\bibliography{aaai24}

\end{document}